\documentclass[onecolumn, english]{IEEEtran}
\usepackage[hidelinks]{hyperref} 
\usepackage{url}   
\usepackage{algorithm}
\usepackage[noend]{algpseudocode}

\algrenewcommand\algorithmicindent{5mm}%

\usepackage{amsmath}
\allowdisplaybreaks
\usepackage{amssymb,amsthm}
\usepackage{amsfonts}
\usepackage{cite}

\usepackage{amssymb,amsthm}
\usepackage{amsfonts}
\usepackage{cite}
\usepackage{tikz}
\usepackage{bbm}
\usetikzlibrary{positioning,arrows,shapes,chains,fit,scopes}
\usepackage{pgfplots}
\usepackage{pgfplotstable}	
\usepackage{booktabs}

\usepackage{threeparttable}
\usepackage{makecell}
\usepackage{multirow}
\usepackage{booktabs,array,longtable}

\usepackage{colortbl}
\usepackage{color}
\definecolor{bgcolor}{rgb}{0.93,0.99,1}
\definecolor{bgcolor2}{rgb}{0.8,1,0.8}
\definecolor{bgcolor3}{rgb}{0.50,0.90,0.50}
\usepackage{tcolorbox}
\usepackage{pifont}
\definecolor{mydarkgreen}{RGB}{39,130,67}
\definecolor{mydarkred}{RGB}{192,25,25}

\usepackage{blkarray, bigstrut}
 
\allowdisplaybreaks

\color{blue}\theoremstyle{plain}\color{black}
\newtheorem{theorem}{Theorem}
\newtheorem{lemma}[theorem]{Lemma}

\newtheorem{proposition}[theorem]{Proposition}
\theoremstyle{definition}

\theoremstyle{remark}
\newtheorem{remark}{Remark}

\usepackage{lipsum}
\usepackage{mathtools}
\usepackage{cuted}

\newcommand{\cX}{\mathcal X}
\newcommand{\cY}{\mathcal Y}
\newcommand{\cZ}{\mathcal Z}

\newcommand{\hbin}{h_2}

\title{Sub-optimality of Marton's Inner Bound for the Two-Receiver Broadcast Channel}
\author{Mian Huang, Yanxiao Liu and Yi Liu
\thanks{The authors are listed in alphabetical order.  

Mian Huang is with the Multimoon Lab, Singapore. Email: \texttt{mhuang5865@gmail.com} 

Yanxiao Liu is with the Department of Electrical and Electronic Engineering, Imperial College London, London, UK. Email: \texttt{y.liu2@imperial.ac.uk}

Yi Liu is with the Department of Information Engineering, The Chinese University of Hong Kong, Hong Kong, China. Email: \texttt{ly023@ie.cuhk.edu.hk}

The code for reproducing our numerical computations is available at \url{https://github.com/yanxiaoliu-mike/Suboptimality_Marton}
}

}

\begin{document}
\maketitle

\begin{abstract}
Marton's inner bound, the best-known achievable region for a general discrete memoryless broadcast channel, was proposed by Katalin Marton in 1979, and whether it always achieves the capacity region has remained open since then.
In this paper, we establish its strict sub-optimality: we show that the capacity region of some discrete memoryless broadcast channels can be strictly larger than Marton's inner bound. 
\end{abstract}

\section{Introduction}

The communication system depicted in Figure \ref{fig:broadcast} presents a two-receiver memoryless broadcast channel $(T_{YZ|X}, \cX, \cY, \cZ)$, where the sender attempts to transmit a private message $M_1$ to receiver $Y$, a private message $M_2$ to receiver $Z$, and a common message $M_0$ to both receivers. 
This model was originally proposed by Cover \cite{cover1972broadcast} (and a detailed definition can be found in \cite[Chapters 5, 8]{el2011network}), but the full  characterization of the capacity region remains one of the most  fundamental open problems in information theory for decades.

\tikzstyle{box}=[rectangle, draw, text centered]
\tikzstyle{line} = [draw, -latex']

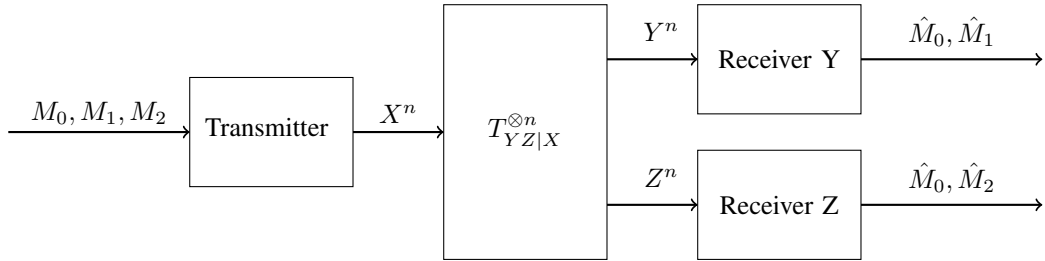
\begin{figure}[htb]
    \centering
    \begin{tikzpicture}[scale=1.2, every node/.style={scale=1}]
\node at (-0.2,-0.6) {$M_0, M_1,M_2$}; 
\draw [->,thick] (-1.2,-0.8) -- (0.8,-0.8);
\draw (0.8,-0.2) rectangle +(1.8,-1.2); \node at (1.65,-0.75) {Transmitter};
\draw [->,thick] (2.6,-0.8)-- (3.6,-0.8); \node at (3.1, -0.6) {$ X^n$};
\draw (3.6, -2.2) rectangle +(1.8,2.8); \node at (4.5,-0.8) {$T_{YZ|X}^{\otimes n}$};
\draw [->,thick] (5.4, 0) --(6.4, 0); \node at (6.0,0.3) {$Y^n$};
\draw [->,thick] (5.4, -1.6) --(6.4, -1.6); \node at (6.0,-1.3) {$Z^n$};
\draw (6.4,-0.6) rectangle +(1.8,1.2); \node at (7.3,0) {Receiver Y};
\draw (6.4,-2.2) rectangle +(1.8,1.2); \node at (7.3,-1.6) {Receiver Z};
\draw [->,thick] (8.2, 0) --(10.2,0); \node at (9.2, 0.3) {$\hat{ M}_0,\hat{ M}_1$};
\draw [->,thick] (8.2, -1.6) --(10.2,-1.6); \node at (9.2, -1.3) {$\hat{ M}_0,\hat{ M}_2$};
\end{tikzpicture}
\caption{Two-receiver broadcast communication system.}
\label{fig:broadcast}
\end{figure}

In 1979, Marton~\cite{Marton1979} proposed the following achievable region, which has since been referred to as \emph{Marton's inner bound} and has remained the best-known achievable region for two-receiver broadcast channels for nearly fifty years. 
\begin{theorem}
[Marton's inner bound, \cite{Marton1979}]
\label{thm:marton}
The union of non-negative rate triples $(R_0, R_1, R_2)$ satisfying the constraints
\begin{subequations}
\begin{align}
    R_0 & \leq \min \{I(W;Y),I(W;Z) \},\label{eqnM1}\\ 
    R_0+R_1 & \leq I(U,W;Y),\label{eqnM2}\\
    R_0+R_2 & \leq I(V,W;Z),\label{eqnM3}\\
    R_0+R_1 + R_2 & \leq \min \{I(W;Y),I(W;Z) \} + I(U;Y|W)+ I(V;Z|W) - I(U;V|W),\label{eqnM4}
\end{align}
\end{subequations}
for any triple of random variables $(U,V,W)$ such that $(U,V,W)\rightarrow  X\rightarrow (Y,Z)$ is achievable for a broadcast channel $T(y,z|x)$. 
We denote this region as $\mathcal{M}(T)$.
\end{theorem}

For many specific classes of broadcast channels, Marton's inner bound matches the capacity region, including \cite{geng2014capacity, ahlswede1975source, bergmans1973random, gamal1979capacity, gallager1974capacity, geng2013marton, nair2009capacity,  nair2016optimality, gel1980capacity, weingarten2006capacity, gohari2026capacity}, where in several of these Marton's inner bound is the capacity: it coincides with the UVW-Outer Bound \cite{nair2011note} for several classes of broadcast channels \cite{ahlswede1975source, bergmans1973random, gamal1979capacity, gallager1974capacity, geng2013marton, nair2009capacity,  nair2016optimality, gel1980capacity, weingarten2006capacity,GengGohariNairYu2014} and coincides with the auxiliary receiver outer bound \cite{gohari2021outer} on classes of sum-broadcast channel \cite{gohari2026capacity}.

A major bottleneck in evaluating Marton's achievable region for the two-receiver broadcast channels is its computability. Early cardinality arguments based on Carath\'eodory-type convexification were ineffective for Marton's mutually dependent auxiliaries. Consequently, evaluating Marton's achievable region---even for binary-input broadcast channels---was intractable using these traditional arguments (see \cite{HajekPursley1979} for an early attempt). 

For a specific binary input channel, the binary skew-symmetric broadcast channel, Nair and Wang \cite{nair08conjecture} conjectured that it suffices to consider either \(U=X\) and \(V\) is constant, or \(V=X\) and \(U\) is constant, to evaluate the sum-rate of Marton's inner bound. Major progress was achieved by Gohari and Anantharam~\cite{gohari2012evaluation}, who introduced a novel perturbation technique. This led to finite cardinality bounds on the auxiliary random variables in Marton's inner bound for any discrete memoryless broadcast channel. Specifically, to compute the full Marton inner bound, one can impose \(|\mathcal{U}| \le |\mathcal{X}|\), \(|\mathcal{V}| \le |\mathcal{X}|\), \(|\mathcal{W}| \le |\mathcal{X}|+4\), and \(H(X|U,V,W)=0\). To compute only the sum-rate, these bounds can be tightened to \(|\mathcal{U}| \le |\mathcal{X}|\), \(|\mathcal{V}| \le |\mathcal{X}|\), \(|\mathcal{W}| \le |\mathcal{X}|+1\), and \(H(X|U,V,W)=0\). For binary-input broadcast channels, this restricts \(U\) and \(V\) to binary alphabets. They also numerically tested the conjecture in \cite{nair08conjecture}.

The above perturbation approach was further developed by Jog and Nair \cite{jog2009} to resolve the conjecture stated in \cite{nair08conjecture}. This was done by showing that $X= U \vee V$ (OR pattern) or $X=U \wedge V$ (AND pattern) could not be local maximizers. This argument was generalized by Geng, Jog, Nair, and Wang \ \cite{NairWangGeng2010, GengJogNairWang2013}, who demonstrated that to compute the sum-rate of Marton's inner bound for any binary-input broadcast channel, either \(U=X\) and \(V\) is constant, or \(V=X\) and \(U\) is constant. This established that randomized time division achieves the Marton sum-rate for binary input broadcast channels. The authors also evaluated the inner and outer bounds for the binary skew-symmetric broadcast channel, establishing the exact gap between them. This discrepancy indicated that either Marton's inner bound, the best-known outer bound at the time, or both, were loose. 

To test the optimality of the Marton sum-rate for binary-input broadcast channels, Geng, Gohari, Nair and Yu proposed numerically comparing two-letter and one-letter Marton's inner bound \cite{geng2011marton}. To evaluate the two-letter Marton's inner bound, one must evaluate it over the product of broadcast channels. However, even for binary-input broadcast channels, computing the two-letter extension using the cardinality bounds from \cite{gohari2012evaluation} yields alphabet sizes of \(4\) for \(U\) and \(V\), and \(5\) for \(W\), resulting in a prohibitively large parameter space. To address this computational hurdle, a sequence of works \cite{NairWangGeng2010, GengJogNairWang2013, GohariNairAnantharam2012, GohariElGamalAnantharam2014} demonstrated that the mappings from \((U,V,W)\) to \(X\) can be vastly restricted; in particular, the so-called (generalized) AND and OR patterns were eliminated. This progress enabled initial simulations comparing the two-letter and one-letter evaluations of Marton's inner bound of binary input broadcast channels for the first time. Extensive simulations conducted between 2011 and 2012 found no examples demonstrating a strict improvement of the two-letter bound over the one-letter bound, leading researchers to conjecture that Marton's inner bound might be tight. This belief was further reinforced by the gradual discovery of the capacity for new classes of channels where Marton's inner bound proved to be tight while the prevailing outer bounds were shown to be loose.

Despite this progress,  numerical calculation of the two-letter and one-letter versions of Marton's inner bound for channels with larger input alphabet size was still practically impossible.
To address this shortcoming, new cardinality bounds were obtained by Anantharam, Gohari and Nair~\cite{AnantharamGohariNair2013,AnantharamGohariNair2019}. These bounds were obtained by combining the perturbation technique~\cite{gohari2012evaluation} with upper concave envelopes~\cite{NairConcave2013}, dual supporting-hyperplane representations, along with properties of optimizers of the inner bound. The initial bounds
$|\mathcal U|\le |\mathcal X|$ and $|\mathcal V|\le |\mathcal X|$ were replaced by $|\mathcal U|+|\mathcal V|\le |\mathcal X|+1$, and the role of the random variable $W$ was expressed through convex envelope.  
These aforementioned bounds made meaningful numerical optimization possible for ternary input broadcast channels. 

As product channels had become a natural avenue for testing the optimality of Marton's inner bound, Geng, Gohari, Nair and Yu in \cite{GengGohariNairYu2014} studied products of different classes of channels such as semi-deterministic, more capable, and less noisy channels. Marton's inner bound was found to be the capacity for such product broadcast channels. In addition, notably, the authors identified two distinct channels \(T_1\) and \(T_2\) such that the Marton sum-rate of \(T_1\otimes T_2\) was strictly larger than the sum of their individual Marton sum-rates. This, however, did not imply that Marton's inner bound is not tight because the channels \(T_1\) and \(T_2\) were different. 

To compare the one-letter and two-letter versions of Marton's inner bound, one can proceed as follows: take the maximizer of the one-letter functional \(p(u,v,w,x)\), take two i.i.d.\ copies \(p(u_1,v_1,w_1,x_1)p(u_2,v_2,w_2,x_2)\), and consider \(U=(U_1,U_2)\), \(V=(V_1,V_2)\), \(W=(W_1,W_2)\) for the two-letter Marton bound. The two-letter Marton bound is equal to the one-letter Marton bound if this product-form distribution is a global maximizer. While it is easy to see that the product distribution satisfies the proper first derivative conditions, if it is a global maximizer, the Hessian at the product distribution must be negative semi-definite. Nair was able to establish the negative semi-definite property for a product broadcast channel when one product component is binary \cite{Nair2020}. A corresponding global tensorization statement would have implied the optimality of Marton's region for binary-input channels. 

In a related approach, Gohari, Liu, and Nair observed that if a global maximizer of Marton's inner bound assumes the product form
\[
p(u_1,v_1,w_1,x_1)p(u_2,v_2,w_2,x_2),
\]
and \(X_i=f_i(U_i,V_i,W_i)\) for \(i=1,2\), then the mapping from \((U_1,U_2), (V_1,V_2), (W_1,W_2)\) to \((X_1,X_2)\) exhibits a ``rectangular form.'' This observation, combined with first-derivative optimality conditions and the implications of the conjectured additivity of Marton's region, led the authors of \cite{gohari2025conjecture} to propose a Markovity/rectangular-mapping conjecture for the optimizers of the dual Marton functional. Specifically, the conjecture posits the existence of a global optimizer for Marton's inner bound that satisfies \(I(U;V|W,X)=0\). This hypothesis was supported by extensive numerical evidence conducted by the authors, encompassing over \(10^4\) channels for each alphabet size \(|\mathcal{X}| \in \{3,4,5\}\), and more than \(100\) channels for each \(|\mathcal{X}| \in \{6,7\}\). Furthermore, although \cite{gohari2025conjecture} reports instances where local optimizers of the Marton expression violate the condition \(I(U;V|W,X)=0\), these local optimizers consistently achieved lower values than the global maximizers. However, as will be demonstrated, the Markovity conjecture does not hold for certain channels. Moreover, these specific channels can be utilized to construct counterexamples to the optimality of Marton's inner bound.

\subsection*{Our Contributions} 

In this paper, we establish the strict sub-optimality of the Marton's inner bound. 
We first identify a ternary input channel instance where the two-letter and one-letter rates, subject to a fixed input distribution,  have a positive gap,  providing a  counterexample to the Marton optimality conjecture under this restriction. The cardinality bound \( |\mathcal U| + |\mathcal V| \le |\mathcal X| + 1 \) implies that, for the evaluation of our ternary-input channels, we only need to consider the following cases: \(U=X\) and \(V\) constant; \(V=X\) and \(U\) constant; or both \(U\) and \(V\) binary. Thanks to these results, one can exhaust all possible distributions and use interval arithmetic to evaluate the one-letter expression. For the two-letter expression, it suffices to exhibit a particular instance of the joint distribution of the auxiliary random variables to establish a lower bound.
We then prove that a new discrete broadcast channel can be constructed for which the inequality persists \emph{without} an input constraint. This latter technique may be of independent interest for other information-theoretic problems. 

Our developments were discovered in the reverse order of the conjectures. We first found a counterexample to the Markovity (rectangular-mapping) conjecture of \cite{gohari2025conjecture}. Then this was used to show that the local tensorization property studied in \cite{Nair2020} failed when both components were ternary. This then led to discovering an example where the global tensorization (additivity conjecture) failed, and finally to examples where  Marton's bound was strictly sub-optimal (first with fixed-input-distributions, and then the unconstrained ones).

We note that earlier researchers had tried essentially the same numerical search that we conducted, but had not succeeded in finding a counterexample to any of the above conjectures because: (i) counterexamples appear to require a non-binary input alphabet, whereas many earlier simulations focused on binary-input channels; (ii) the counterexamples occupy only a small region of the search space; for instance, counterexamples to the Markovity conjecture did not appear in over \(10^4\) random channels experiment of \cite{gohari2025conjecture}; and (iii) the violation gaps are small, of order \(10^{-6}\). The AI-guided search algorithm used by the first author was based on  the method of elimination geometry \cite{Huang2026eg} and was iterative and structure-driven: it successively used counterexamples for one conjecture to guide the search for counterexamples to another, progressing from the failure of the Markovity conjecture to dual nonadditivity, and then to a fixed-input-distribution version of Marton’s conjecture.

The paper is organized as follows.
After introducing the necessary notation and preliminaries, we introduce our main techniques in Section~\ref{sec:tools}, including a gradient-shaping tool and a constraint-removal argument for constructing an unconstrained counterexample from an input-constrained counterexample.
In Section~\ref{sec:counterexamples}, we prove the strict sub-optimality of Marton's inner bound, starting from an input-constrained counterexample and utilizing the aforementioned tools.

\subsection*{Notation and preliminaries}
The full Marton's inner bound has been introduced in Theorem~\ref{thm:marton}. 
In this paper we use the full auxiliary $W$ but focus on the private-message case $R_0=0$, since a strict improvement in the maximum private-message sum-rate is already enough to prove strict sub-optimality of the complete Marton region.

Given a broadcast channel $T_{YZ|X}$, consider any distribution $P_{WUVX}$ satisfying $(W,U,V)\rightarrow X\rightarrow (Y,Z)$. Define the information functional
\begin{equation}
M_T(P_{WUVX})
:=\min\{I(W;Y),I(W;Z)\}
+I(U;Y|W)+I(V;Z|W)-I(U;V|W).
\label{eq:MTlaw}
\end{equation}
For a fixed input distribution $P_X$, let
\begin{equation*}
M_T(P_X):=\max_{P_{WUV|X}} M_T(P_{WUVX}),
\end{equation*}
The quantity $M_T(P_{X})$ equals the largest sum-rate permitted by Marton's inner bound under the prescribed input distribution $P_X$. 

Accordingly, Marton's sum-rate $M_T$ for the channel $T$ can be written as:
$$M_T = \max_{P_{WUVX}} M_T(P_{WUVX})=\max_{P_X} M_T(P_X).$$

For a prescribed joint input distribution $q_{X_1X_2}$, the quantity $M_{T^{\otimes 2}}(q_{X_1X_2})$ is defined analogously for the product channel $T^{\otimes 2}$, without normalization by a factor of two.

For each $\alpha\in[0,1]$, define the affine functional
\begin{equation}
L_{\alpha,T}(P_{WUVX})
:=(1-\alpha)I(W;Y)+\alpha I(W;Z)
+I(U;Y|W)+I(V;Z|W)-I(U;V|W).
\label{eq:Ltdef}
\end{equation}
For every fixed distribution $P_{WUVX}$, we have 
\[
M_T(P_{WUVX})
=\min_{\alpha \in[0,1]}L_{\alpha,T}(P_{WUVX})
=\min\{L_{0,T}(P_{WUVX}),L_{1,T}(P_{WUVX})\},
\]
where the second equality follows from the fact that $L_{\alpha,T}(P_{WUVX})$ is affine in $\alpha$.  

In a similar vein, for a fixed input distribution $P_X$, define
\[ L_{\alpha,T}(P_X) := \max_{p_{WUV|X}} L_{\alpha,T}(P_{WUVX}),\]
and let $$ L_{\alpha,T} := \max_{P_X} L_{\alpha,T}(P_X).$$ 
The max-min theorem established in \cite{GengGohariNairYu2014} yields
\begin{align*}
    M_T(P_X) &= \min_{\alpha \in [0,1]} L_{\alpha,T}(P_X), \\ 
    M_T &= \min_{\alpha \in [0,1]} L_{\alpha,T}. 
\end{align*}

Finally, for any $\alpha\in [0,1]$ and any sequence $\{a_x\}\in \mathbb R^{|\cX|},x\in \cX$, define the dual functional
$$F_T(\{a_x\},\alpha) := \max_{P_{UVX}}G_T(P_{UVX},\{a_x\},\alpha),$$
where
\begin{equation}
    G_T(P_{UVX},\{a_x\},\alpha) := 
    -\alpha H(Y)-(1-\alpha) H(Z) +I(U;Y)+I(V;Z)-I(U;V) +\sum_{x\in\cX}p(x)a_x.
    \label{eq:markovity-F}
\end{equation}

Throughout the paper all logarithms are natural, so all rates are measured in nats unless explicitly stated otherwise. 
All numerical inequalities used for the explicit instance are certified by exact rational arithmetic and outward-rounded MPFR calculations.

\section{Main Techniques}
\label{sec:tools}

In this section, we first introduce the techniques we use to construct a counterexample to Marton's optimality. Note that if $M_{T^{\otimes 2}}>2M_T$, then a normalized two-letter Marton scheme achieves a rate pair whose sum \emph{strictly} exceeds the sum-rate of every point in the one-letter Marton region. 
Consequently, the one-letter Marton region is strictly sub-optimal. 
In order to construct such a counterexample, we first study the following question:

\textbf{Sub-optimality of Marton at a fixed input distribution:}
We say that Marton's inner bound is not tight at a fixed input distribution $p^*_X$ if one can find $q^*_{X_1X_2}$ such that for any $x$, 
\begin{equation}
p_X^*(x)
=
\frac12
\bigl(
q^*_{X_1}(x)+q^*_{X_2}(x)
\bigr), 
\label{eq:average-marginal}
\end{equation}
and $M_{T^{\otimes2}}(q^*_{X_1X_2})>2M_{T}(p_X^*)$.
\begin{remark}
We will present such a channel and $p^*_X$ in Section \ref{sec:counterexamples}.
\end{remark}

\subsection{From a constrained-input counterexample to an unconstrained counterexample}
We now introduce a technique to construct an unconstrained counterexample from an input-constrained broadcast-channel instance. 
Suppose that a gap between the two-letter and one-letter rates has been found at a prescribed input distribution $p^*$, but that $p^*$ does not maximize the corresponding unconstrained one-letter objective. 
Simply removing the input constraint may then destroy the separation, since the one-letter optimizer is free to move to another input distribution. 
A \emph{gradient shaping} lemma provides an almost automatic device for modifying a concave fixed-input objective so that the prescribed distribution $p^*$ becomes a global maximizer. 
The second construction, \emph{constraint removal}, realizes this gradient-shaping device through an enlarged broadcast channel and shows that the original multiletter gain can be preserved after the input constraint is removed. Theorem \ref{thm:constraint-removal} below then utilizes this concept to formally state a result for moving from a constrained problem to an unconstrained problem.

We begin with the purely geometric gradient-shaping problem: given a finite concave function $f$ on the probability simplex, we seek a family of concave perturbations that makes a prescribed point $p^*$ a global maximizer. 
The observation is that, for each coordinate $i$, the concave function $p\mapsto  \hbin(p_i)+p_iL$ has partial-derivative at $p^*$ equal to $\big(
L+\log\frac{1-p_i^*}{p_i^*} \big)$. Thus, these perturbations allow us to modify the gradient independently, one coordinate at a time.

\begin{theorem}[Gradient shaping]
\label{thm:gradient_shaping}
Let $f$ be a finite concave function on the probability simplex, let $p^*=(p_1^*,\ldots,p_d^*)$ lie in its relative interior, and let $g=(g_1,\ldots,g_d)$ be any supergradient of $f$ at $p^*$. Choose $r\in\arg\max_i g_i$. 
For sufficiently large $L$, for each $i\neq r$, define $A_i(L):=
L+\log\frac{1-p_i^*}{p_i^*}$ and
\begin{equation}
\lambda_L
:=
\Big(
1+\sum_{i\neq r}\frac{g_r-g_i}{A_i(L)}
\Big)^{-1},
\qquad
\delta_{i,L}
:=
\lambda_L\frac{g_r-g_i}{A_i(L)},
\end{equation}
both of which are nonnegative, and then we have $\lambda_L+\sum_{i\neq r}\delta_{i,L}=1$, and the concave function
\begin{equation}
\Phi_L(p)
:=
\lambda_L f(p)
+
\sum_{i\neq r}\delta_{i,L}
\bigl[\hbin(p_i)+p_iL\bigr]
\label{eq:Phiabstract}
\end{equation}
has $p^*$ as a global maximizer. 
Moreover, $\lambda_L\to 1$ and $\delta_{i,L}\to0$ as $L\to\infty$. 
\end{theorem}

The proof can be found in Appendix~\ref{app:gradient_shaping}.

For a broadcast channel $T$ instance whose multiletter rate has a strict gap, Theorem~\ref{thm:gradient_shaping} now provides a way to ``perturb'' a prescribed input distribution $p^*$ so that it becomes the global maximizer. 
The perturbations $\hbin(p_i)+p_iL$ arise naturally in broadcast-channel settings, as shown below.

Fix an integer $N\ge2$, set $L=\log N$, and enlarge the broadcast channel to
\[
\widetilde X=(X,K)\in\cX\times[N].
\]
For each input symbol $i\in\cX$, define
\begin{equation}
G_i(X,K)
:=
\begin{cases}
K,&X=i,\\
e_i,&X\neq i,
\end{cases}
\label{eq:Gi-realization}
\end{equation}
where $e_i\notin[N]$ is a distinct erasure symbol corresponding to each $i$. 
We refer to this as a \emph{common deterministic component}, since both receivers will be given the same deterministic output $G_i(X,K)$. Note that $E_i := 1_{X=i} $ is determined by $G_i$.
If the $X$-marginal is $p$, then for an arbitrary conditional distribution of $K$,
\begin{equation}
H(G_i)
= H(G_i,E_i) = H(E_i) + P(X=i) H(G_i|E_i=1) =
\hbin(p_i)+p_iH(K|X=i)
\le
\hbin(p_i)+p_iL.
\label{eq:Gi-psi}
\end{equation}
Equality holds when $K$ is uniform on $[N]$ conditional on $X=i$.  Thus the coordinate perturbation in Theorem~\ref{thm:gradient_shaping} is the largest entropy that this common-output component can generate at a prescribed $X$-marginal.

It remains to realize the \emph{weighted sum} of these perturbations together
with the original broadcast channel.  Recall that
Theorem~\ref{thm:gradient_shaping} chooses a reference coordinate $r$ and
weights $\lambda_L$ and $\{\delta_{i,L}:i\neq r\}$ satisfying $\lambda_L+\sum_{i\neq r}\delta_{i,L}=1$. 
We therefore construct a new broadcast channel $\widetilde T_L$ as follows. 
The new channel can be viewed as a state-dependent broadcast channel where the state variables are known to the two receivers. The state variable is denoted by $Q$.

Given an input $(x,k)\in\cX\times[N]$, the broadcast channel independently draws a \emph{state}, which is independent of the input variables
$(X,K)$, by
\[
Q\in\{0\}\cup\{i\in\cX:i\neq r\}
\]
with $\Pr[Q=0]=\lambda_L$ and $\Pr[Q=i]=\delta_{i,L}$ for $i\neq r$.

If $Q=0$, the channel uses the original broadcast channel $T$ on the $X$-coordinate and ignores $K$: it draws $(Y,Z)\sim T(\cdot,\cdot|x)$ and gives the receivers
\[
\widetilde Y=(0,Y),
\qquad
\widetilde Z=(0,Z).
\]
If $Q=i\neq r$, the channel instead uses the common deterministic component $G_i$ and gives both receivers
\[
\widetilde Y=(i,G_i(x,k)),
\qquad
\widetilde Z=(i,G_i(x,k)).
\] 
Thus the value of $Q$ is included explicitly in both receiver outputs. After conditioning on $Q$, the mutual-information terms decompose into the corresponding weighted contributions of the original channel and the shaping components. The following proposition makes this decomposition precise. For notational convenience in the proposition below, let $\mathcal I:=\{i\in\cX:i\neq r\}$ and write $\lambda:=\lambda_L$ and $\delta_i:=\delta_{i,L}$ for $i\in\mathcal I$.

\begin{proposition}
\label{prop:channel-shaper}
Consider the receiver-revealed mixture defined above and fix a joint distribution $P_{WUV\widetilde X}$. 
Assume that $X$ is a function of $(U,V,W)$. (It suffices to consider such distributions). As defined earlier, $M_T(P_{UVWX})$ denotes the Marton functional obtained from this distribution when the base channel $T$ is applied to the input $X$ and $K$ is ignored, and let $M_{G_i}(P_{UVWX})$ denote the Marton functional when both receivers observe $G_i$. 
Then
\begin{equation}
M_{\widetilde T_L}(P_{UVWX})
=
\lambda M_T(P_{UVWX})
+
\sum_{i\in\mathcal I}\delta_i M_{G_i}(P_{UVWX}).
\label{eq:revealed-mixture}
\end{equation}
The same identity holds with $M$ replaced by every affine form $L_\alpha$.
Moreover, if both receivers observe the same deterministic output $G=g(\widetilde X)$ and the input distribution of $\widetilde X$ is $\pi$, then
\begin{equation}
M_{\widetilde T_L}(P_{UVWX})
\leq 
\lambda M_T(P_{UVWX})
+
\sum_{i\in\mathcal I}\delta_i (H_{G_i}(P_{UVWX}) - I(U;V|W,G_i)).
\label{eq:common-deterministic}
\end{equation}
\end{proposition}

The proof can be found in Appendix \ref{app:proof_channel_shaper}. 

Now let $f$ be any finite concave function satisfying $M_T(P_X)\le f(p)$ for any $p\in\Delta(\cX)$, and let $p^*$ be a point at which we want to control the one-letter optimizer.  We can apply Theorem~\ref{thm:gradient_shaping} to $f$ at $p^*$, choose the reference coordinate $r$ and the weights $\lambda_L,\delta_{i,L}$ given there, and form the receiver-revealed mixture with $\mathcal I=\operatorname{supp}(p^*)\setminus\{r\}$.

For any distribution $P_{WUVX}$ on the enlarged channel whose marginal $P_X =p$, Proposition~\ref{prop:channel-shaper} and \eqref{eq:Gi-psi} give
\begin{align}
M_{\widetilde T_L}(P_{UVWX})
&\le
\lambda_L M_T(P_X)
+
\sum_{i\neq r}\delta_{i,L}
\bigl[\hbin(p_i)+p_iL\bigr]\nonumber\\
&\le
\lambda_L f(p)
+
\sum_{i\neq r}\delta_{i,L}
\bigl[\hbin(p_i)+p_iL\bigr]
=
\Phi_L(p).
\label{eq:channel-Phi}
\end{align}
Theorem~\ref{thm:gradient_shaping} makes $p^*$ a global maximizer of $\Phi_L$, hence our channel construction prevents the single-letter Marton's sum-rate from exceeding $\Phi_L(p^*)$. 
The remaining question is whether this modification preserves a multiletter rate advantage that was originally available at $p^*$. 
We answer this question affirmatively by the following theorem.

\begin{theorem}[Constraint removal]
\label{thm:constraint-removal}
Let $T$ be a finite broadcast channel.  
Let $f:\Delta(\cX)\to\mathbb R$ be finite and concave satisfying $M_T(P_X)\le f(p)$ for any $p\in\Delta(\cX)$. 
Suppose there exist an input distribution $p^*$ for $T$ and a joint distribution $P^*_{UVWX_1X_2}$ for $T^{\otimes2}$, with input marginal $q^*_{X_1X_2}$, such that for any $x$, 
\begin{equation}
p^*(x)
=
\frac12
\bigl(
q^*_{X_1}(x)+q^*_{X_2}(x)
\bigr), 
\label{eq:average-marginal}
\end{equation}
and $M_*^{(2)}:=M_{T^{\otimes2}}(P^*)>2f(p^*)$, then there exists a finite broadcast channel $T'$ with no input constraint such that
\begin{equation*}
M_{T'^{\otimes2}}>2M_{T'}.
\end{equation*}
In particular, since $p\mapsto M_T(P_X)$ is concave, if $M_T(p^*)<\frac12M_{T^{\otimes2}}(q^*_{X_1X_2})$, then the same conclusion holds.
\end{theorem}

The proof of Theorem \ref{thm:constraint-removal} can be found in Appendix \ref{app:constraint-removal}. 
Note $1-\lambda_L=O(L^{-1})$ and $\delta_{i,L}=O(L^{-1})$, so the shaping components are used with vanishing probability, while $\delta_{i,L}A_i(L)=\lambda_L(g_r-g_i)$ remains finite and changes the supporting slope by the required amount. 
Theorem~\ref{thm:gradient_shaping} itself is independent of Marton's inner bound and may be of independent interest for other information theory problems.

\section{Sub-optimality of Marton's Inner Bound}
\label{sec:counterexamples}

We now formally prove the strict sub-optimality of the Marton's inner bound, i.e., a counterexample to the Marton optimality conjecture. 
We first present a ternary ``base channel'' with a strict two-letter gain under a fixed input distribution. 
This channel instance was searched by computer assistance, and then we apply the techniques developed in the previous section to construct an unconstrained counterexample.

\subsection{A Ternary Base Channel Violating the Markovity Conjecture}
Consider a ternary-input ternary-output broadcast channel with $\cX=\cY=\cZ=\{0,1,2\}$. Its two receiver marginals are given, to the stated numerical precision, by
\begin{equation}
T_Y=
\begin{pmatrix}
0.385361043864&0.358548005177&0.256090950959\\
0.790302113033&0.091025394735&0.118672492232\\
0.113168779849&0.520687239901&0.366143980250
\end{pmatrix},
\label{eq:TYseed}
\end{equation}
and
\begin{equation}
T_Z=
\begin{pmatrix}
0.390396602690&0.207292393713&0.402311003597\\
0.293627666338&0.040944965497&0.665427368165\\
0.396880111705&0.507497881167&0.095622007128
\end{pmatrix}.
\label{eq:TZseed}
\end{equation} 
Consider the input distribution
\begin{equation}
p^*=(0.704504835923,0.172563088231,0.122932075846).
\label{eq:pstar-seed}
\end{equation}

We first show that this channel provides a counterexample to the Markovity conjecture proposed in \cite[Conjecture 2]{gohari2025conjecture}. The conjecture asserts that, for every broadcast channel $T_{YZ|X}$, every $\alpha\in(0,1)$, and every sequence $\{a_x\}\in \mathbb R^{|\cX|},x\in \cX$, there exists a distribution $P_{UVX}$ attaining $F_T(\{a_x\},\alpha)$ such that $U\rightarrow X \rightarrow V$, or, equivalently, such that $I(U;V|X)=0$. 

Set
\begin{equation}
\label{eq:base-markovity-alpha}
    \alpha_0=0.546388044486, \qquad a=(0,-0.393210549796,-0.149290191142). 
\end{equation}
Numerical optimization gives
\begin{align}
F_T(\{a_x\},\alpha)
\simeq
-1.06471340944277385029433493729151639713942682792468927190596.
\end{align}
The optimizing distribution is specified by
\begin{equation}
X=f(U,V)
=
\begin{pmatrix}
0&1\\
2&0
\end{pmatrix}, \qquad
P_{UV} 
=
\begin{pmatrix}
0.013502132210&0.172563088231\\
0.122932075846&0.691002703713
\end{pmatrix}.
\label{eq:ase-markovity-PUV}
\end{equation}
Here, the rows and columns correspond to the values of \(U\) and \(V\),
respectively. The induced input distribution is precisely \(p^*\) in
\eqref{eq:pstar-seed}. Moreover, the deterministic mapping \(f\) is
non-rectangular, and the resulting distribution does not satisfy the
Markov chain \(U\to X\to V\). Indeed,
\[
I(U;V|X)\simeq 0.06676808506814671821263044747>0.
\]

In contrast, restricting the optimization to distributions satisfying
\(U\to X\to V\) yields
\begin{equation}
\max_{P_{UVX}:\,U\rightarrow X \rightarrow V}
G_T(P_{UVX},\{a_x\},\alpha) <-1.0648<F_T(\{a_x\},\alpha).
\label{eq:markovity-separation}
\end{equation}
The gap between the unrestricted and Markovian optima is approximately
\begin{equation}
8.74201195024\times 10^{-5}\ \text{nats}.
\end{equation}

Consequently, no distribution satisfying \(U\to X\to V\) can attain the
unrestricted optimum, and the channel in
\eqref{eq:TYseed}--\eqref{eq:TZseed} violates the Markovity conjecture.

For additional counterexamples obtained independently at an early stage of this investigation, see \cite{liu2026counterexamples}.

\subsection{A Fixed-Input Counterexample to Marton Optimality}
We next show that the ternary broadcast channel introduced above provides a fixed-input counterexample to the optimality of the one-letter Marton's inner bound. 

We use the same parameters as in \eqref{eq:base-markovity-alpha}, namely,
\begin{equation}
\alpha_0=0.546388044486,
\qquad
\{a_x\}=(0,-0.393210549796,-0.149290191142). 
\label{eq:alphaa}
\end{equation}

Following the formulation in \cite{gohari2025conjecture}, we first derive an upper bound on $M_T(P_X)$. Define $$B_T(\alpha,\{a_x\}) :=\max_{q(u,v,x)}-\alpha H(Y)-\bar \alpha H(Z)+I(U;Y)+I(V;Z)-I(U;V)+\sum_{x\in\cX}q(x)a_x.$$

For any prescribed input distribution \(P_X\), define
\begin{equation}
U_B^{\alpha,\{a_x\}}(P_X)
:=
\alpha H_p(Y)+\bar{\alpha}H_p(Z)
+B_T(\alpha,\{a_x\})
-\sum_{x\in\cX}p(x)a_x. 
\label{eq:UB-definition}
\end{equation}

We claim that
\begin{equation}
M_T(P_X)
\leq L_{\alpha,T}(P_X)
\leq U_B^{\alpha,\{a_x\}}(P_X).
\label{eq:UB-claim}
\end{equation}

To establish the second inequality, let $P_{WUV|X}$ attain $L_{\alpha,T}(P_X)$. Then
\begin{align*}
    L_{\alpha,T}(P_X)=&L_{\alpha,T}(P_{WUVX})\\
    =&(1-\alpha)I(W;Y)+\alpha I(W;Z)+I(U;Y|W)+I(V;Z|W)-I(U;V|W)\\
    =&\alpha H(Y)+\bar\alpha H(Z)-\sum_{x\in\cX}p(x)a_x+\sum_{W=w} P(W=w)[-\alpha H(Y|W=w)-\bar \alpha H(Z|W=w)\\
    &+I(U;Y|W=w)+I(V;Z|W=w)-I(U;V|W=w) +\sum_{x\in\cX}p(x|W=w)a_x]\\
    \leq &\alpha H(Y)+\bar\alpha H(Z)-\sum_{x\in\cX}p(x)a_x+\bigg(\max_{q(u,v,x)}-\alpha H(Y)-\bar \alpha H(Z)+I(U;Y)+I(V;Z)-I(U;V)+\sum_{x\in\cX}q(x)a_x\bigg) \\
    = &\alpha H(Y)+\bar\alpha H(Z)+B_T(\alpha,\{a_x\})-\sum_{x\in\cX}p(x)a_x\\
    =&U_B^{\alpha,\{a_x\}}(P_X),
\end{align*}
where the inequality follows because conditioning on $W=w$ induces an input distribution $p(x|w)$, and the resulting conditional objective is an average over $w$. This average cannot exceed the maximum of the corresponding unconditional objective over all admissible input distributions and auxiliary random variables.

A numerical evaluation shows that,
\begin{equation}
U_B^{\alpha,\{a_x\}}(p^*)\le 0.0916508833913554746710, 
\label{eq:Ucert}
\end{equation}

On the two-letter side, we can find an explicit rational input distribution $(P^*)_{WUVX_1X_2}$ for $T^{\otimes2}$ satisfying
\begin{equation}
(P^*)_{X_1}=(P^*)_{X_2}=p^*,
\label{eq:same-marginals}
\end{equation}
and

\begin{equation}
M_{T^{\otimes2}}(P^*)
\geq
0.1833045943206552249020> 2U_B^{\alpha,\{a_x\}}(p^*)
\geq
2M_T(p^*).
\label{eq:Vcert}
\end{equation}

We next describe the two-letter distribution $P^*$ explicitly. 
For $q\in\{0,1\}$, let $P_q=P_{W_cUVX_1X_2\mid Q=q}$ be the distribution specified by the table in Appendix~\ref{app:twoletter-component-laws}. 
The alphabets satisfy $|\mathcal W_c|=2$ and $|\mathcal U|=|\mathcal V|=4$.
Let $Q$ be binary with $P_Q(1)=1-\mu$ and $P_Q(0)=\mu\approx 0.28950945961978169531$, and, conditional on $Q=q$, draw $(W_c,U,V,X_1,X_2)$ according to $P_q$.  Moreover, take $W=(Q,W_c)$. 
This defines the joint distribution of $P^*_{WUVX_1X_2}$ denoted by $P^*$ above and its input marginal: 
\begin{equation}
P^*_{X_1X_2}
\approx
\begin{pmatrix}
0.504171085044430 & 0.118736798891440 & 0.081596951987130\\
0.118736798891440 & 0.030668142751315 & 0.023158146588245\\
0.081596951987130 & 0.023158146588245 & 0.018176977270626
\end{pmatrix}.
\label{eq:qstar-input}
\end{equation}
More details on the numerical calculation of $M_{T^{\otimes2}}(P^*)$ can also be found in Appendix~\ref{app:twoletter-component-laws}.  

Consequently, the gap is at least  $2.82753794427556\times10^{-6} >0$. 
Therefore the base channel has a strict fixed-input two-letter gain at $p^*$. 
It is not difficult to verify the numerical gain once the channel distribution and $p^*$ are fixed. 
It also yields a cost-constrained counterexample to the tensorization (see \cite[Conjecture 1]{gohari2025conjecture}), and hence to the optimality of the one-letter Marton inner bound.
 
\subsection{From the base channel to an unconstrained counterexample}

The fixed-input-distribution gain above is exactly the situation addressed by Theorem~\ref{thm:constraint-removal}.  
Indeed, \eqref{eq:same-marginals} implies
\[
p^*
=
\frac12\bigl((P^*)_{X_1}+(P^*)_{X_2}\bigr),
\]
while \eqref{eq:Ucert} and \eqref{eq:Vcert} give
\[
M_{T^{\otimes2}}(P^*)>2U_B^{\alpha,\{a_x\}}(p^*).
\]
Therefore Theorem~\ref{thm:constraint-removal}, applied with the concave envelope $f=U_B^{\alpha,\{a_x\}}$, directly gives a finite unconstrained broadcast channel $\widetilde T$ satisfying
\begin{equation}
M_{\widetilde T^{\otimes2}}>2M_{\widetilde T}.
\label{eq:main-subopt}
\end{equation}
The existence argument does not require fixing a particular label size: in the proof of Theorem~\ref{thm:constraint-removal}, one lets $L=\log N$ tend to infinity.  The shaping probabilities $\delta_{i,L}$ then tend to zero, while the original positive gap survives.

\begin{theorem}[Sub-optimality of Marton's inner bound]
\label{thm:main-sub-optimality}
There exists a finite two-receiver discrete memoryless broadcast channel $\widetilde T$ for which
\[
M_{\widetilde T^{\otimes2}}>2M_{\widetilde T}.
\]
Consequently, the complete one-letter Marton region is strictly contained in the capacity region of $\widetilde T$.
\end{theorem}

\subsection{Numerical Example}

For completeness and numerical reproducibility, we next give one fully explicit finite realization of the channel guaranteed by Theorem~\ref{thm:constraint-removal}.  
We choose
\begin{equation}
M=2,000,000,
\qquad
N=2^M,
\qquad
L=M\log2.
\label{eq:explicit-N}
\end{equation}
The enlarged input alphabet is
\begin{equation}
\widetilde\cX=\{0,1,2\}\times[N],
\qquad
\widetilde X=(X,K).
\label{eq:explicit-input}
\end{equation}
Since the largest coordinate of the certified supporting gradient is the coordinate corresponding to $x=2$, the two shaping components act on $x=0$ and $x=1$.  We use the exact rational probabilities
\begin{equation}
\lambda=\frac{99999977198527428838942}{10^{23}},
\qquad
\delta_0=\frac{20970790337007940}{10^{23}},
\qquad
\delta_1=\frac{1830682234153118}{10^{23}},
\label{eq:explicit-weights}
\end{equation}
and define
\begin{equation}
G_j(x,k)
=
\begin{cases}
k,&x=j,\\
e_j,&x\neq j,
\end{cases}
\qquad j=0,1,
\label{eq:explicit-g}
\end{equation}
where $e_0,e_1\notin[N]$ are erasure symbols.

The channel $\widetilde T$ is the receiver-revealed mixture of Proposition~\ref{prop:channel-shaper}: given $(x,k)$, with probability $\lambda$ it uses the base channel $T$ on $x$ and ignores $k$; with probability $\delta_j$, $j=0,1$, both receivers observe the common deterministic output $G_j(x,k)$.  To specify a joint broadcast channel in the base component, we use the conditionally independent coupling of the two marginals.  Equivalently, the only nonzero joint transitions are
\begin{align}
\widetilde T((0,y),(0,z)|(x,k))
&=\lambda T_Y(y|x)T_Z(z|x),
\label{eq:kernelB}\\
\widetilde T((j+1,G_j(x,k)),(j+1,G_j(x,k))|(x,k))
&=\delta_j,
\qquad j=0,1.
\label{eq:kernelS}
\end{align} 
The tags $0,1,2$ are part of both receiver outputs, so the selected component is revealed to both receivers.  This is an ordinary finite memoryless broadcast channel with no input cost, fixed-composition constraint, or encoder-side restriction.

We now verify the one-letter side using exactly the gradient-shaping mechanism of the previous section. 
We have
\begin{equation}
b=(0,b_1,b_2),
\qquad
b_1=0.26533808894270830145,
\qquad
b_2=0.29071676800873404161,
\label{eq:explicit-b}
\end{equation}
together with a scalar $\nu$ and a residual vector $\rho$ such that
\begin{equation}
\nabla U_B(p^*)
=
b+\nu\mathbf 1+\rho,
\qquad
\max_i\rho_i-\min_i\rho_i
<
4.692109048907095\times10^{-22}.
\label{eq:base-gradient}
\end{equation}
Since $b_2>b_1>b_0$, the reference coordinate in Theorem~\ref{thm:gradient_shaping} is $r=2$.  The rational weights in \eqref{eq:explicit-weights} are chosen as close rational approximations to the corresponding gradient-shaping weights.  A $320$-bit outward-rounded calculation verifies that the complete nonconstant gradient range of
\begin{equation}
\Phi(p)
:=
\lambda U_B(p)
+
\delta_0[\hbin(p_0)+p_0L]
+
\delta_1[\hbin(p_1)+p_1L]
\label{eq:explicit-Phi}
\end{equation}
at $p^*$ is at most
\begin{equation}
\varepsilon_\nabla
<
1.846843880121467\times10^{-17}.
\label{eq:graderr}
\end{equation}
By concavity,
\[
\Phi(p)\le \Phi(p^*)+\varepsilon_\nabla
\qquad\forall p\in\Delta(\cX).
\]
Using \eqref{eq:Ucert}, Proposition~\ref{prop:channel-shaper}, and directed rounding therefore gives
\begin{equation}
M_{\widetilde T}
\le
\lambda U
+
\delta_0[\hbin(p_0^*)+p_0^*L]
+
\delta_1[\hbin(p_1^*)+p_1^*L]
+
\varepsilon_\nabla
\le
0.30084186609915448.
\label{eq:explicit-U}
\end{equation}

For the two-letter lower bound, use the exact rational base-channel distribution $P^*$, draw $K_1,K_2$ independently and uniformly on $[N]$, independently of all variables in $P^*$, and put the labels into the common auxiliary:
\[
\widetilde W=(W,K_1,K_2),
\qquad
\widetilde X_i=(X_i,K_i),
\quad i=1,2.
\]
On the base-base selector pair, every affine Marton form is at least $M_*^{(2)}$.  If coordinate $i$ uses shaping component $j$, the output reveals $K_i$ on $\{X_i=j\}$ and hence contributes exactly $p_j^*L$ to both common-message terms.  On all non-base selector pairs, we discard the remaining nonnegative conditional mutual-information terms and use
\[
I(U;V|\widetilde W)\le\log4,
\]
since the base-channel witness has $|U|=|V|=4$.  Therefore, uniformly for every $\alpha\in[0,1]$,
\begin{equation}
L_{\alpha,\widetilde T^{\otimes2}}
\ge
\lambda^2M_*^{(2)}
+
2L(\delta_0p_0^*+\delta_1p_1^*)
-
(1-\lambda^2)\log4.
\label{eq:explicit-two-affine}
\end{equation}
Taking the minimum over $\alpha$ gives
\begin{equation}
M_{\widetilde T^{\otimes2}}
\ge
\lambda^2M_*^{(2)}
+
2L(\delta_0p_0^*+\delta_1p_1^*)
-
(1-\lambda^2)\log4
\ge
0.60168561432394163.
\label{eq:explicit-two-lower}
\end{equation}
Combining \eqref{eq:explicit-U} and \eqref{eq:explicit-two-lower} yields
\begin{equation}
M_{\widetilde T^{\otimes2}}-2M_{\widetilde T}
>
1.8821256327185490\times10^{-6}.
\label{eq:explicit-gap}
\end{equation} 

Consequently, the one-letter Marton region of this channel instance is strictly contained in its capacity region. 

\section{Concluding Remarks}

In this paper, we establish the strict sub-optimality of Marton's one-letter inner bound for the two-receiver discrete memoryless broadcast channel.  
Starting from a certified fixed-input two-letter gain, we introduce two general techniques, \emph{gradient shaping} and \emph{constraint removal}, to convert the constrained separation into an unconstrained one. 
Applying this framework to the ternary example yields a finite broadcast channel satisfying $M_{T^{\otimes2}}>2M_T$, thereby proving that the complete one-letter Marton region can be strictly sub-optimal. 
We also used the same channel instance to show the Markovity conjecture in \cite{gohari2025conjecture} is untrue as well. 
Beyond the applications in this paper, the gradient-shaping mechanism may provide a useful general device for converting fixed-input multiletter separations into unconstrained ones in other information theory problems. 
It remains interesting to see if Marton's inner bound is tight for binary input broadcast channels.

\section{Acknowledgement and Bibliographic Remarks}

The authors would like to thank Prof. Chandra Nair and Prof. Amin Gohari for their valuable suggestions on the counterexample search, their verification of the technical non-numerical aspects of the paper, and their generous assistance with the paper's organization and writing. 

The numerical search relied heavily on the assistance of AI models, which is acknowledged together with the bibliographic records as follows. 

Chandra Nair requested Yanxiao Liu to search for counterexamples to the additivity conjecture (\cite[Conjecture~1]{gohari2025conjecture}) using AI during their meeting at ISIT 2026.
With assistance from GPT-5.6 Sol, Yanxiao Liu first discovered a counterexample to the Markovity conjecture (\cite[Conjecture~2]{gohari2025conjecture}).
Independently, by the method of elimination geometry \cite{Huang2026eg}, Mian Huang also discovered a counterexample to the Markovity conjecture, assisted by GPT-5.6 Sol. 
The two authors were then put in touch with each other after they both privately communicated their results to Amin Gohari, Yi Liu, and Chandra Nair, and the two counterexamples were summarized in \cite{liu2026counterexamples}. 
Yi Liu then informed them about the local tensorization test~\cite{Nair2020}.
Mian Huang found a counterexample to the local tensorization test for a broadcast channel with both components being ternary, motivated by the counterexample to the failed Markovity conjecture, assisted by  Claude Fable 5.
This was then developed into a counterexample to the additivity conjecture.
Amin Gohari and Chandra Nair then suggested that they search for a counterexample while constraining themselves to a fixed input distribution.
Such an example was found by Mian Huang, assisted by Claude Fable 5 and Opus 5. 
Yanxiao Liu then modified this example to construct an unconstrained counterexample to Marton's region, assisted by GPT-5.6 Sol.

\appendices

\section{Exact two-letter component distributions}
\label{app:twoletter-component-laws}
The two rational component distributions used for $P_q=P_{W_cUVX_1X_2\mid Q=q}$ are specified below.  Each row gives a nonzero value of
$P_q(w_c,u,v,x_1,x_2)$; the listed numerator is divided by $10^{12}$, and every entry not listed is zero.  The pair $x_1x_2$ is written as a two-digit string. 
\begingroup
\normalsize
\setlength{\tabcolsep}{6pt}
\begin{longtable}{ccccc r}
\hline
$q$ & $w_c$ & $u$ & $v$ & $x_1x_2$ & numerator\\
\hline
\endfirsthead
\hline
$q$ & $w_c$ & $u$ & $v$ & $x_1x_2$ & numerator\\
\hline
\endhead
0&0&0&0&\texttt{11}&$2$\\
0&0&0&1&\texttt{01}&$1{,}384{,}478{,}680$\\
0&0&0&2&\texttt{10}&$1{,}384{,}478{,}680$\\
0&0&0&3&\texttt{11}&$15{,}902{,}736{,}192$\\
0&0&1&0&\texttt{10}&$16$\\
0&0&1&1&\texttt{00}&$10{,}209{,}427{,}117$\\
0&0&1&2&\texttt{12}&$23{,}193{,}325{,}149$\\
0&0&1&3&\texttt{10}&$117{,}244{,}957{,}226$\\
0&0&2&0&\texttt{01}&$16$\\
0&0&2&1&\texttt{21}&$23{,}193{,}325{,}149$\\
0&0&2&2&\texttt{00}&$10{,}209{,}427{,}117$\\
0&0&2&3&\texttt{01}&$117{,}244{,}957{,}226$\\
0&0&3&0&\texttt{00}&$55$\\
0&0&3&1&\texttt{20}&$81{,}553{,}873{,}546$\\
0&0&3&2&\texttt{02}&$81{,}553{,}873{,}546$\\
0&0&3&3&\texttt{00}&$411{,}998{,}751{,}962$\\
0&1&0&0&\texttt{00}&$732{,}352{,}219$\\
0&1&0&3&\texttt{11}&$14{,}837{,}590{,}966$\\
0&1&1&0&\texttt{22}&$42$\\
0&1&1&3&\texttt{00}&$166$\\
0&1&2&0&\texttt{22}&$18$\\
0&1&2&3&\texttt{00}&$70$\\
0&1&3&0&\texttt{22}&$18{,}184{,}877{,}091$\\
0&1&3&3&\texttt{00}&$71{,}171{,}567{,}749$\\
\hline
1&0&0&0&\texttt{11}&$3$\\
1&0&0&1&\texttt{01}&$1{,}422{,}792{,}872$\\
1&0&0&2&\texttt{10}&$1{,}422{,}792{,}872$\\
1&0&0&3&\texttt{11}&$16{,}393{,}806{,}593$\\
1&0&1&0&\texttt{10}&$25$\\
1&0&1&1&\texttt{00}&$10{,}187{,}677{,}513$\\
1&0&1&2&\texttt{12}&$23{,}143{,}812{,}089$\\
1&0&1&3&\texttt{10}&$117{,}357{,}754{,}072$\\
1&0&2&0&\texttt{01}&$25$\\
1&0&2&1&\texttt{21}&$23{,}143{,}812{,}089$\\
1&0&2&2&\texttt{00}&$10{,}187{,}677{,}513$\\
1&0&2&3&\texttt{01}&$117{,}357{,}754{,}072$\\
1&0&3&0&\texttt{00}&$88$\\
1&0&3&1&\texttt{20}&$81{,}614{,}505{,}516$\\
1&0&3&2&\texttt{02}&$81{,}614{,}505{,}516$\\
1&0&3&3&\texttt{00}&$413{,}588{,}568{,}366$\\
1&1&0&0&\texttt{00}&$720{,}151{,}199$\\
1&1&0&2&\texttt{11}&$18$\\
1&1&0&3&\texttt{11}&$14{,}244{,}922{,}559$\\
1&1&1&0&\texttt{22}&$45$\\
1&1&1&3&\texttt{00}&$171$\\
1&1&2&0&\texttt{22}&$26$\\
1&1&2&3&\texttt{00}&$100$\\
1&1&3&0&\texttt{22}&$18{,}173{,}758{,}170$\\
1&1&3&1&\texttt{00}&$1$\\
1&1&3&2&\texttt{00}&$88$\\
1&1&3&3&\texttt{00}&$69{,}425{,}708{,}399$\\
\hline
\end{longtable}
\endgroup 

We also provide more details on then numerical calculation of $M_{T^{\otimes2}}(P^*)$: after defining $q\in\{0,1\}$, we have \begin{align*}
J_q
&:=I_q(U;Y^2\mid W_c)+I_q(V;Z^2\mid W_c)-I_q(U;V\mid W_c),\\
A_q
&:=I_q(W_c;Y^2)+J_q=L_{0,T^{\otimes2}}(P_q),\\
B_q
&:=I_q(W_c;Z^2)+J_q=L_{1,T^{\otimes2}}(P_q).
\end{align*} 
Directed-rounding evaluation of the two exact rational component distributions gives the rigorous lower bounds $A_q\ge\underline A_q$ and $B_q\ge\underline B_q$, where
\begin{equation}
\begin{array}{c|cc}
q&\underline A_q&\underline B_q\\ \hline
0&0.1833138563360490321659&0.1832960014540807003610\\
1&0.1833008202506018554789&0.1833080957271021625141
\end{array}
\label{eq:twoletter-endpoint-bounds}
\end{equation}
The exact rational value of $\mu$ is chosen so that the two mixtures of these certified endpoint lower bounds agree:
\begin{equation}
\mu \underline A_0+(1-\mu)\underline A_1
=
\mu \underline B_0+(1-\mu)\underline B_1
=:
\gamma,
\label{eq:twoletter-equalize}
\end{equation}
where
\begin{equation}
\gamma =0.183304594320655224901767\ldots.
\label{eq:twoletter-V}
\end{equation}
Indeed, using $W=(Q,W_c)$ and the chain rule,
\begin{align}
I(W;Y^2)
&=I(Q;Y^2)+\mu I_0(W_c;Y^2)+(1-\mu)I_1(W_c;Y^2),\label{eq:twoletter-chain-Y}\\
I(W;Z^2)
&=I(Q;Z^2)+\mu I_0(W_c;Z^2)+(1-\mu)I_1(W_c;Z^2).\label{eq:twoletter-chain-Z}
\end{align}
The conditional Marton terms equal $\mu J_0+(1-\mu)J_1$.  Dropping the nonnegative terms $I(Q;Y^2)$ and $I(Q;Z^2)$ therefore gives
\begin{align}
M_{T^{\otimes2}}(P^*)
&\ge
\min\bigl\{
\mu A_0+(1-\mu)A_1,
\mu B_0+(1-\mu)B_1
\bigr\}\nonumber\\
&\ge
\min\bigl\{
\mu \underline A_0+(1-\mu)\underline A_1,
\mu \underline B_0+(1-\mu)\underline B_1
\bigr\}\nonumber\\
&=\gamma.
\label{eq:twoletter-direct-lower}
\end{align}

\section{Properties}

\begin{lemma}[Concavity]
\label{lem:concavity}
For every finite broadcast channel $T$, the function $p\mapsto M_T(P_X)$ is concave on $\Delta(\cX)$.
\end{lemma}

\begin{proof}
Let $p^{(0)},p^{(1)}$ be two input distributions and let $P^{(0)}_{WUVX}$, $P^{(1)}_{WUVX}$ be feasible joint distributions with these marginals.  Let $Q\sim\mathrm{Bernoulli}(\theta)$ select the two distributions, and put $W'=(Q,W_Q)$ while embedding $U_Q,V_Q$ into common alphabets.  The input marginal is $p=\theta p^{(0)}+(1-\theta)p^{(1)}$.  The three conditional terms in \eqref{eq:MTlaw} average exactly.  Moreover
\[
I(W';Y)\ge \theta I(W_0;Y)+(1-\theta)I(W_1;Y),
\]
and the same holds for $Z$.  Hence
\begin{align*}
\min\{I(W';Y),I(W';Z)\}
&\ge \min\bigl\{\theta a_0+(1-\theta)a_1,\theta b_0+(1-\theta)b_1\bigr\}\\
&\ge \theta\min\{a_0,b_0\}+(1-\theta)\min\{a_1,b_1\},
\end{align*}
where $a_q=I(W_q;Y)$ and $b_q=I(W_q;Z)$.  Thus the time-shared Marton value is at least the corresponding convex combination.  Taking suprema proves concavity.
\end{proof}

\section{Proof of Theorem \ref{thm:gradient_shaping}}
\label{app:gradient_shaping}
\begin{proof}
For sufficiently large $L$, every $A_i(L)$ is positive.  Since $g$ is a supergradient of $f$ at $p^*$ and $A_i(L)\vec{e}_i$ is the gradient at $p^*$ of the function $p\mapsto \hbin(p_i)+p_iL$, the vector
\[
\lambda_L g
+
\sum_{i\neq r}\delta_{i,L}A_i(L)\vec{e}_i
\]
is a supergradient of $\Phi_L$ at $p^*$. Here $\vec{e}_i$ denotes the elementary basis vector in the Euclidean space.  Its $r$th coordinate is $\lambda_Lg_r$.  For $i\neq r$, its $i$th coordinate is
\[
\lambda_Lg_i+\delta_{i,L}A_i(L)
=
\lambda_Lg_i+\lambda_L(g_r-g_i)
=
\lambda_Lg_r.
\]
Hence this supergradient is the constant vector $\lambda_Lg_r\mathbf 1$.  For every probability vector $p$,
\[
\Phi_L(p)
\le
\Phi_L(p^*)
+
\lambda_Lg_r\sum_i(p_i-p_i^*)
=
\Phi_L(p^*),
\]
so $p^*$ is a global maximizer.  Finally, $A_i(L)=L+O(1)$, hence
$\delta_{i,L}=O(L^{-1})$ and $\lambda_L=1-O(L^{-1})$.
\end{proof}

\section{Proof of Proposition \ref{prop:channel-shaper}}
\label{app:proof_channel_shaper}
\begin{proof}
For the first claim, write $Y_0=Y$, $Z_0=Z$, and $Y_i=Z_i=G_i$ for $i\in\mathcal I$.  Since $Q$ is independent of $(W,U,V,\widetilde X)$ and is included in both receiver outputs,
\begin{align*}
I(W;\widetilde Y)
&=
\lambda I(W;Y)
+
\sum_{i\in\mathcal I}\delta_i I(W;G_i),\\
I(W;\widetilde Z)
&=
\lambda I(W;Z)
+
\sum_{i\in\mathcal I}\delta_i I(W;G_i).
\end{align*}
The two conditional mutual informations decompose in the same way.  Because every shaping component contributes the same common-information term to the two receivers,
\begin{align*}
\min\{I(W;\widetilde Y),I(W;\widetilde Z)\}
={}&
\lambda\min\{I(W;Y),I(W;Z)\}\\
&+
\sum_{i\in\mathcal I}\delta_i I(W;G_i).
\end{align*}
Finally,
\[
-I(U;V|W)
=
-\left(\lambda+\sum_{i\in\mathcal I}\delta_i\right)I(U;V|W),
\]
so collecting the contribution of each component gives \eqref{eq:revealed-mixture}.  The argument for $L_t$ is identical, without taking the minimum of the two common-message terms.

For the second claim, for every joint distribution with input marginal $\pi$,
\begin{align*}
M_G(P)
&=
I(W;G)+I(U;G|W)+I(V;G|W)-I(U;V|W)\\
&=
I(W,U,V;G)-I(U;V|W,G)\\
&\le H_\pi(G).
\end{align*}
Equality is achieved by taking $W=G$ and $U,V$ constant. 
Consequently, the unconstrained Marton sum-rate of this component is $\max_\pi H_\pi(G)=\log|g(\widetilde\cX)|$.  
Since both receivers observe exactly the same deterministic output, the same quantity is also its true private-message sum-capacity.
\end{proof}

\section{Proof of Theorem \ref{thm:constraint-removal}}
\label{app:constraint-removal}
\begin{proof}
Restrict the input alphabet to $\operatorname{supp}(p^*)$.  By \eqref{eq:average-marginal}, neither coordinate of the seed distribution uses a symbol outside this support.  Hence $p^*$ lies in the relative interior of the reduced simplex.

Let $g$ be a supergradient of $f$ at $p^*$ and choose
$r\in\arg\max_i g_i$.  For sufficiently large $L=\log N$, choose
$\lambda_L$ and $\delta_{i,L}$ as in Theorem~\ref{thm:gradient_shaping}, and construct the receiver-revealed mixture $\widetilde T_L$ described above.  By \eqref{eq:channel-Phi} and Theorem~\ref{thm:gradient_shaping},
\begin{equation}
M_{\widetilde T_L}
\le
\lambda_L f(p^*)
+
\sum_{i\neq r}\delta_{i,L}
\bigl[\hbin(p_i^*)+p_i^*L\bigr].
\label{eq:oneletter-general-final}
\end{equation}

For the two-letter lower bound, use the seed distribution $P^*$ and draw
$K_1,K_2$ independently and uniformly from $[N]$, independently of all seed variables.  Put these labels into the common auxiliary:
\[
\widetilde W=(W,K_1,K_2),
\qquad
\widetilde X_\ell=(X_\ell,K_\ell),
\qquad \ell=1,2.
\]
Let
\[
C:=I_{P^*}(U;V|W).
\]
Condition on the two independently drawn selectors of the two copies of $\widetilde T_L$ and evaluate the affine Marton functional $L_t$.  On the event that both selectors choose the base channel, which has probability $\lambda_L^2$, the contribution is
$L_{t,T^{\otimes2}}(P^*)$.

On every other selector pair, discard all nonnegative finite conditional mutual-information terms.  If coordinate $\ell$ selects shaping component $i$, its output reveals $K_\ell$ exactly on the event $\{X_\ell=i\}$.  Since $K_\ell$ is uniform and independent of the seed distribution, this contributes
$q^*_{X_\ell}(i)L$ to each of the two common-message terms.  The negative term $-C$ is present for every selector pair.  Consequently, uniformly for every $t\in[0,1]$,
\begin{align}
L_{t,\widetilde T_L^{\otimes2}}
&\ge
\lambda_L^2L_{t,T^{\otimes2}}(P^*)
+
L\sum_{i\neq r}\delta_{i,L}
\bigl(q^*_{X_1}(i)+q^*_{X_2}(i)\bigr)
-
(1-\lambda_L^2)C
\nonumber\\
&=
\lambda_L^2L_{t,T^{\otimes2}}(P^*)
+
2L\sum_{i\neq r}\delta_{i,L}p_i^*
-
(1-\lambda_L^2)C,
\label{eq:twoletter-general}
\end{align}
where the second equality follows from \eqref{eq:average-marginal}.  Taking the minimum over $t$ gives
\begin{equation}
M_{\widetilde T_L^{\otimes2}}
\ge
\lambda_L^2M_*^{(2)}
+
2L\sum_{i\neq r}\delta_{i,L}p_i^*
-
(1-\lambda_L^2)C.
\label{eq:twoletter-general-final}
\end{equation}

Subtracting twice \eqref{eq:oneletter-general-final}, the leading label terms cancel exactly:
\begin{align}
M_{\widetilde T_L^{\otimes2}}
-
2M_{\widetilde T_L}
\ge{}&
\lambda_L^2M_*^{(2)}
-
2\lambda_L f(p^*)
-
2\sum_{i\neq r}\delta_{i,L}\hbin(p_i^*)
-
(1-\lambda_L^2)C.
\label{eq:gap-general}
\end{align}
By Theorem~\ref{thm:gradient_shaping},
\[
\lambda_L\to1,
\qquad
\delta_{i,L}\to0.
\]
Therefore the right-hand side of \eqref{eq:gap-general} converges to
\[
M_*^{(2)}-2f(p^*)>0.
\]
It is thus strictly positive for all sufficiently large finite $N$.  Choosing such an $N$ and setting $T'=\widetilde T_L$ proves $M_{T'^{\otimes2}}>2M_{T'}$. 

For the final statement, take $f=M_T$, which is concave by Lemma~\ref{lem:concavity}.  
If $M_T(p^*)<\frac12M_{T^{\otimes2}}(q^*_{X_1X_2})$ holds, the strict inequality allows us to choose a finite feasible two-letter distribution $P^*$ with input marginal $q^*_{X_1X_2}$ and
\[
M_{T^{\otimes2}}(P^*)>2M_T(p^*).
\]
The preceding argument then applies.
\end{proof}

\bibliographystyle{IEEEtran}
\bibliography{ref.bib}

\begin{thebibliography}{10}
\providecommand{\url}[1]{#1}
\csname url@samestyle\endcsname
\providecommand{\newblock}{\relax}
\providecommand{\bibinfo}[2]{#2}
\providecommand{\BIBentrySTDinterwordspacing}{\spaceskip=0pt\relax}
\providecommand{\BIBentryALTinterwordstretchfactor}{4}
\providecommand{\BIBentryALTinterwordspacing}{\spaceskip=\fontdimen2\font plus
\BIBentryALTinterwordstretchfactor\fontdimen3\font minus \fontdimen4\font\relax}
\providecommand{\BIBforeignlanguage}[2]{{%
\expandafter\ifx\csname l@#1\endcsname\relax
\typeout{** WARNING: IEEEtran.bst: No hyphenation pattern has been}%
\typeout{** loaded for the language `#1'. Using the pattern for}%
\typeout{** the default language instead.}%
\else
\language=\csname l@#1\endcsname
\fi
#2}}
\providecommand{\BIBdecl}{\relax}
\BIBdecl

\bibitem{cover1972broadcast}
T.~Cover, ``Broadcast channels,'' \emph{IEEE Transactions on Information Theory}, vol.~18, no.~1, pp. 2--14, 1972.

\bibitem{el2011network}
A.~El~Gamal and Y.-H. Kim, \emph{Network information theory}.\hskip 1em plus 0.5em minus 0.4em\relax Cambridge university press Cambridge, UK, 2011.

\bibitem{Marton1979}
K.~Marton, ``A coding theorem for the discrete memoryless broadcast channel,'' \emph{IEEE Transactions on Information Theory}, vol.~25, no.~3, pp. 306--311, 1979.

\bibitem{geng2014capacity}
Y.~Geng and C.~Nair, ``The capacity region of the two-receiver gaussian vector broadcast channel with private and common messages,'' \emph{IEEE Transactions on Information Theory}, vol.~60, no.~4, pp. 2087--2104, 2014.

\bibitem{ahlswede1975source}
R.~Ahlswede and J.~Korner, ``Source coding with side information and a converse for degraded broadcast channels,'' \emph{IEEE Transactions on Information Theory}, vol.~21, no.~6, pp. 629--637, 1975.

\bibitem{bergmans1973random}
P.~Bergmans, ``Random coding theorem for broadcast channels with degraded components,'' \emph{IEEE Transactions on Information Theory}, vol.~19, no.~2, pp. 197--207, 1973.

\bibitem{gamal1979capacity}
A.~Gamal, ``The capacity of a class of broadcast channels,'' \emph{IEEE Transactions on Information Theory}, vol.~25, no.~2, pp. 166--169, 1979.

\bibitem{gallager1974capacity}
R.~G. Gallager, ``Capacity and coding for degraded broadcast channels,'' \emph{Problemy Peredachi Informatsii}, vol.~10, no.~3, pp. 3--14, 1974.

\bibitem{geng2013marton}
Y.~Geng, A.~Gohari, C.~Nair, and Y.~Yu, ``On {Marton}'s inner bound and its optimality for classes of product broadcast channels,'' \emph{IEEE Transactions on Information Theory}, vol.~60, no.~1, pp. 22--41, 2013.

\bibitem{nair2009capacity}
C.~Nair, ``Capacity regions of two new classes of 2-receiver broadcast channels,'' in \emph{2009 IEEE International Symposium on Information Theory}.\hskip 1em plus 0.5em minus 0.4em\relax IEEE, 2009, pp. 1839--1843.

\bibitem{nair2016optimality}
C.~Nair, H.~Kim, and A.~El~Gamal, ``On the optimality of randomized time division and superposition coding for the broadcast channel,'' in \emph{2016 IEEE Information Theory Workshop (ITW)}.\hskip 1em plus 0.5em minus 0.4em\relax IEEE, 2016, pp. 131--135.

\bibitem{gel1980capacity}
S.~I. Gel'fand and M.~S. Pinsker, ``Capacity of a broadcast channel with one deterministic component,'' \emph{Problemy Peredachi Informatsii}, vol.~16, no.~1, pp. 24--34, 1980.

\bibitem{weingarten2006capacity}
H.~Weingarten, Y.~Steinberg, and S.~S. Shamai, ``The capacity region of the gaussian multiple-input multiple-output broadcast channel,'' \emph{IEEE Transactions on Information Theory}, vol.~52, no.~9, pp. 3936--3964, 2006.

\bibitem{gohari2026capacity}
A.~Gohari, Y.~Liu, and C.~Nair, ``The capacity region for classes of sum-broadcast channels,'' in \emph{2025 IEEE International Symposium on Information Theory (ISIT)}.\hskip 1em plus 0.5em minus 0.4em\relax IEEE, 2026, pp. 1--6.

\bibitem{nair2011note}
C.~Nair, ``A note on outer bounds for broadcast channel,'' \emph{arXiv preprint arXiv:1101.0640}, 2011.

\bibitem{GengGohariNairYu2014}
Y.~Geng, A.~Gohari, C.~Nair, and Y.~Yu, ``The capacity region of classes of product broadcast channels,'' \emph{IEEE Transactions on Information Theory}, vol.~60, no.~1, pp. 22--41, 2014.

\bibitem{gohari2021outer}
A.~Gohari and C.~Nair, ``Outer bounds for multiuser settings: The auxiliary receiver approach,'' \emph{IEEE Transactions on Information Theory}, vol.~68, no.~2, pp. 701--736, 2021.

\bibitem{HajekPursley1979}
B.~Hajek and M.~B. Pursley, ``Evaluation of an achievable rate region for the broadcast channel,'' \emph{IEEE Transactions on Information Theory}, vol.~25, no.~1, pp. 36--46, 1979.

\bibitem{nair08conjecture}
C.~Nair and Z.~V. Wang, ``On the inner and outer bounds for 2-receiver discrete memoryless broadcast channels,'' in \emph{Proceedings of the Information Theory and Applications Workshop (ITA)}, 2008, pp. 226--229.

\bibitem{gohari2012evaluation}
A.~Gohari and V.~Anantharam, ``Evaluation of {Marton}'s inner bound for the general broadcast channel,'' \emph{IEEE Transactions on Information Theory}, vol.~58, no.~2, pp. 608--619, 2012.

\bibitem{jog2009}
V.~Jog and C.~Nair, ``An information inequality for the {BSSC} channel,'' in \emph{Proceedings of the Information Theory and Applications Workshop (ITA)}, 2010.

\bibitem{NairWangGeng2010}
C.~Nair, Z.~V. Wang, and Y.~Geng, ``An information inequality and evaluation of {Marton}'s inner bound for binary input broadcast channels,'' in \emph{Proceedings of the IEEE International Symposium on Information Theory (ISIT)}, 2010, pp. 550--554.

\bibitem{GengJogNairWang2013}
Y.~Geng, V.~Jog, C.~Nair, and Z.~V. Wang, ``An information inequality and evaluation of {Marton}'s inner bound for binary input broadcast channels,'' \emph{IEEE Transactions on Information Theory}, vol.~59, no.~7, pp. 4095--4105, 2013.

\bibitem{geng2011marton}
Y.~Geng, A.~Gohari, C.~Nair, and Y.~Yu, ``On {Marton}’s inner bound for two receiver broadcast channels,'' in \emph{Proceedings of the Information Theory and Applications Workshop (ITA)}, 2011.

\bibitem{GohariNairAnantharam2012}
A.~Gohari, C.~Nair, and V.~Anantharam, ``On {Marton}'s inner bound for broadcast channels,'' in \emph{Proceedings of the IEEE International Symposium on Information Theory (ISIT)}, 2012, pp. 581--585, full version available as arXiv:1202.0898.

\bibitem{GohariElGamalAnantharam2014}
A.~Gohari, A.~El~Gamal, and V.~Anantharam, ``On {Marton}'s inner bound for the general broadcast channel,'' \emph{IEEE Transactions on Information Theory}, vol.~60, no.~7, pp. 3748--3762, 2014.

\bibitem{AnantharamGohariNair2013}
V.~Anantharam, A.~Gohari, and C.~Nair, ``Improved cardinality bounds on the auxiliary random variables in {Marton}'s inner bound,'' in \emph{Proceedings of the IEEE International Symposium on Information Theory (ISIT)}, 2013, pp. 1272--1276.

\bibitem{AnantharamGohariNair2019}
------, ``On the evaluation of {Marton}'s inner bound for two-receiver broadcast channels,'' \emph{IEEE Transactions on Information Theory}, vol.~65, no.~3, pp. 1361--1371, 2019.

\bibitem{NairConcave2013}
C.~Nair, ``Upper concave envelopes and auxiliary random variables,'' \emph{International Journal of Advances in Engineering Sciences and Applied Mathematics}, vol.~5, no.~1, pp. 12--20, 2013.

\bibitem{Nair2020}
------, ``On {Marton}'s achievable region: Local tensorization for product channels with a binary component,'' in \emph{Proceedings of the Information Theory and Applications Workshop (ITA)}, 2020, pp. 1--7.

\bibitem{gohari2025conjecture}
A.~Gohari, Y.~Liu, and C.~Nair, ``A conjecture regarding the optimizers of {Marton}'s inner bound for the two-receiver broadcast channel,'' in \emph{2025 IEEE International Symposium on Information Theory (ISIT)}.\hskip 1em plus 0.5em minus 0.4em\relax IEEE, 2025, pp. 1--6.

\bibitem{Huang2026eg}
M.~Huang and X.~Wang, ``Elimination geometry - from local optima to global realizability: Structural realizability, certification, and repair in statistical learning and ai,'' 2026, research monograph, August 2026.

\bibitem{liu2026counterexamples}
Y.~Liu and M.~Huang, ``Counterexamples to the {Markovity} conjecture for the two-receiver broadcast channel,'' \emph{arXiv preprint arXiv:2608.13170}, 2026.

\end{thebibliography}

\end{document}